\documentclass[11pt]{llncs}

\usepackage{makeidx}
\usepackage{amsmath}
\usepackage{amssymb}
\usepackage{enumerate}
\usepackage{fullpage}

\begin{document}

\def\urltilda{\kern -.15em\lower .7ex\hbox{\~{}}\kern .04em}
\def\urldot{\kern -.10em.\kern -.10em}
\def\urlhttp{http\kern -.10em\lower -.1ex\hbox{:}\kern -.12em\lower 0ex\hbox{/}\kern -.18em\lower 0ex\hbox{/}}

\mainmatter

\title{Gradual sub-lattice reduction and a new complexity for factoring polynomials}

\titlerunning{Gradual sub-lattice reduction}

\author{Mark van Hoeij\inst{1}\thanks{Supported by NSF 0728853} \and Andrew Novocin\inst{2}}

\authorrunning{Hoeij and Novocin}   

\institute{Florida State University, 208 Love Building Tallahassee, FL 32306-4510 \\ \texttt{hoeij@math\urldot fsu\urldot edu} -- \texttt{\urlhttp www\urldot math\urldot fsu\urldot edu/\urltilda hoeij} \\
\and
LIP/INRIA/ENS, 46 all\'ee d'Italie, F-69364 Lyon Cedex 07, France \\ \texttt{Andrew\urldot Novocin@ens-lyon\urldot fr} -- \texttt{\urlhttp andy\urldot novocin\urldot com/pro} }



\maketitle

\begin{abstract}
  \noindent We present a lattice algorithm specifically designed for some classical applications of lattice reduction.  The applications are for lattice bases with a generalized knapsack-type structure, where the target vectors are boundably short.  For such applications, the complexity of the algorithm improves traditional lattice reduction by replacing some dependence on the bit-length of the input vectors by some dependence on the bound for the output vectors.  If the bit-length of the target vectors is unrelated to the bit-length of the input, then our algorithm is only linear in the bit-length of the input entries, which is an improvement over the quadratic complexity floating-point LLL algorithms.  To illustrate the usefulness of this algorithm we show that a direct application to factoring univariate polynomials over the integers leads to the first complexity bound improvement since 1984.  A second application is algebraic number reconstruction, where a new complexity bound is obtained as well.
\end{abstract}

\makeatletter
\def\Ddots{\mathinner{\mkern1mu\raise\p@
\vbox{\kern7\p@\hbox{.}}\mkern2mu
\raise4\p@\hbox{.}\mkern2mu\raise7\p@\hbox{.}\mkern1mu}}
\makeatother

\newtheorem{ass}{Assumption}
\newtheorem{alg}{Algorithm}
\newcommand{\rv}{n_{\textrm{rm}}}

\section{Introduction}\label{intro}

\noindent Lattice reduction algorithms are essential tools in computational number theory and cryptography.  A lattice is a discrete subset of $\mathbb{R}^n$ that is also a $\mathbb{Z}$-module.  The goal of lattice reduction is to find a `nice' basis for a lattice, one which is near orthogonal and composed of short vectors.  Since the publication of the 1982 Lenstra, Lenstra, Lov\'asz~\cite{LLL} lattice reduction algorithm many applications have been discovered, such as polynomial factorization~\cite{LLL,Hoeij} and attacking several important public-key cryptosystems including knapsack cryptosystems~\cite{knapsack}, RSA under certain settings~\cite{rsa}, and DSA and some signature schemes in particular settings~\cite{dsa}.  One of the important features of the LLL algorithm was that it could approximate the shortest vector of a lattice in polynomial time.  This is valuable because finding the exact shortest vector in a lattice is provably NP-hard~\cite{ajtai,svp}.  Given a basis ${\bf b}_ 1, \ldots, {\bf b}_ d \in \mathbb{R}^n$ which satisfies $\parallel {\bf b}_{i} \parallel \leq X ~~\forall i$, the LLL algorithm has a running time of $\mathcal{O}(d^5 n \log^3{X})$ using classical arithmetic.  Recently there has been a resurgence of lattice reduction work thanks to Nguyen and Stehl\'e's ${\rm L}^2$ algorithm~\cite{fpLLL,fpLLL2} which performs lattice reduction in $\mathcal{O}(d^4 n \log{X} [d+\log{X}])$ CPU operations.  The primary result of ${\rm L}^2$ was that the dependence on $\log{X}$ is only quadratic allowing for improvement on applications using large input vectors.

\noindent \textbf{The main result:}  Many applications of LLL (see the applications section below) involve finding a vector in a lattice whose norm is known to be small in advance.  In such cases it can be more efficient to reduce a basis of a sub-lattice which contains all targeted vectors than reducing a basis of the entire lattice.  In this paper we target short vectors in specific types of input lattice bases which we call knapsack-type bases.  The new algorithm introduces a search parameter $B$ which the user provides.  This parameter is used to bound the norms of targeted short vectors.  To be precise: 

\vspace{.1in}

The {\em rows} of the following matrices represent a knapsack-type basis 

\begin{center}
$ \left( \begin{array}{ccc|ccc}
0 &  \cdots & 0 & 0     & \cdots & P_N\\
0 &  \cdots & 0 & 0     & \Ddots & 0\\
0 &  \cdots & 0 & P_1   & \cdots & 0\\
\hline 
1 & \cdots & 0 & x_{1,1}  &\cdots & x_{1,N}\\ 
\vdots &  \ddots & \vdots & \vdots &   & \vdots \\ 
0 & \cdots & 1 &  x_{r,1} &\cdots & x_{r,N}\\
\end{array} \right)$
\textrm{ or }
$ \left( \begin{array}{ccc|ccc}
1 &  \cdots & 0 & x_{1,1} &\cdots & x_{1,N}\\ 
\vdots & \ddots & \vdots & \vdots &   & \vdots \\ 
0 & \cdots & 1 &  x_{r,1} &\cdots & x_{r,N}\\
\end{array} \right)$.
\end{center}

\noindent The specifications of our algorithm are as follows.  It takes as input a knapsack-type basis ${\bf b}_ 1, \ldots, {\bf b}_ d \in \mathbb{Z}^{n}$ of a lattice $L$ with $\parallel {\bf b}_{i} \parallel \leq X$ $\forall i$ and a search parameter $B$; it returns a \em{reduced} basis generating a \em{sub-lattice} $L' \subseteq L$ such that if ${\bf v} \in L$ and $\parallel {\bf v} \parallel \leq B$ then $ {\bf v} \in L'$.

\noindent Our algorithm has the following complexity bounds for various input:

\vspace{.1in}

\begin{center}
\begin{tabular}{|c|c|}
\hline
No $P_i$ &\raisebox{11pt}{~}$\mathcal{O}(d^2(n+d^2)(d+\log{B})[\log{X}+n(d+\log{B})])$  \\

\hline

No restriction on $P_i$ &\raisebox{11pt}{~}$\mathcal{O}(d^4(d+\log{B})[\log{X}+d(d+\log{B})])$  \\

\hline

Many $P_i$ large w.r.t. $B$ &\raisebox{11pt}{~} $\mathcal{O}(dr^3(r+\log{B})[\log{X}+d(r+\log{B})])$\\

\hline
\end{tabular}
\end{center}

\vspace{.1in}

\noindent These complexity bounds have several distinct parameters, so a comparison with other algorithms is a bit subtle.  The most significant parameter to explore is $B$, the search parameter.  If one selects $B = X$ then our algorithm will return a reduced basis of $L'=L$ in $\mathcal{O}(d^2n(n+d^2)[d^2+ \log^2{X}])$.  This is an interesting result because our algorithm, like the original LLL and the ${\rm L}^2$ algorithms, uses switches and size-reductions of the vectors to arrive at a reduced basis.  The fact that we return a reduced basis with a complexity so similar to ${\rm L}^2$ implies that there are alternative orderings on the switches which lead to similar performance.

\noindent When using a smaller value of $B$ than $X$ the algorithm will return either:
\begin{itemize}
\item A reduced basis of a sub-lattice $L'$ which contains all vectors of norm $\leq B$.  This sub-lattice may be different than the sub-lattice, $L''$, generated by all vectors of norm $\leq B$, and we do have $L'' \subseteq L' \subseteq L$.  Also, because the basis of $L'$ is reduced, we have an approximation of the shortest non-zero vector of $L$. 

\item The empty set, in which case the algorithm has proved that no non-zero vector of norm $\leq B$ exists in $L$.

\end{itemize}

\noindent We offer the following complexity comparison with ${\rm L}^2$~\cite{fpLLL} for some values of $B$ on square input lattices (with $P_j$'s).  When a column has a non-zero $P_j$ we can reduce the $x_{i,j}$ modulo $P_j$.  Thus, without loss of generality, we may assume that $P_j$ is the largest element in its column.  Note that $r=d-N$.

\begin{center} 
\begin{tabular}{|c|c|}
\hline
${\rm L}^2$ & \raisebox{11pt}{~}$\mathcal{O}(d^6\log{X}+d^5\log^2{X})$\\
\hline
\hline
$B=\mathcal{O}(X)$ & \raisebox{11pt}{~}$\mathcal{O}(d^7+d^5\log^2{X})$\\
\hline
$B=\mathcal{O}(X^{1/d})$ &\raisebox{11pt}{~} $\mathcal{O}(d^2r^5+r^3\log^2{X})$\\
\hline
$B=2^{\mathcal{O}(d)}$ &\raisebox{11pt}{~} $\mathcal{O}(d^4r^3+d^2r^3\log{X})$ \\
\hline
\end{tabular}
\end{center}

\noindent It should be noted that~\cite{fpLLL} explores running times of ${\rm L}^2$ on knapsack lattices with $N=1$ (such lattice bases are used in~\cite{goldstein}).  In this case, ${\rm L}^2$ will have complexity $\mathcal{O}(d^5\log{X}+d^4\log^2{X})$.

\noindent \textbf{Our approach:}  We reduce the basis gradually, using many separate calls to another lattice reduction algorithm.  To get the above complexity results we chose H-LLL~\cite{HLLL} but there are many suitable lattice reduction algorithms we could use instead such as~\cite{kaltofen,LLL,fpLLL,schnorr2}.  For more details on why we made this decision see the discussion in section~\ref{compsect}.  

There are three important features to our approach.  First, we approach the problem column by column.  Beginning with the $r\times r$ identity and with each iteration of the algorithm we expand our scope to include one more column of the $x_{i,j}$.  Next, within each column iteration, we reduce the new entries bit by bit, starting with a reduction using only the most significant bits, then gradual including more and more bits of data.  Third, we allow for the removal of vectors which have become too large.  This allows us to always work on small entries, but restricts us to a sub-lattice.  

The proof of the algorithm's complexity is essentially a study of two quantities, the product of the Gram-Schmidt lengths of the current vectors which we call the active determinant and an energy function which we call progress.  We amortize all of the lattice reduction costs using progress, and we bound the number of iterations and number of vectors using the active determinant.  Neither of these quantities is impacted by the choice of lattice reduction algorithm.

\noindent \textbf{Applications of the algorithm:} As evidence for the usefulness of this new approach we show two new complexity results based on applications of the main algorithm.  The first result is a new complexity for the classical problem of factoring polynomials in $\mathbb{Z}[x]$.  If the polynomial has degree $N$, coefficients smaller than $\log(A)$, and when reduced modulo a prime $p$ has $r$ irreducible factors then we prove a complexity of $\mathcal{O}(N^3r^4+N^2r^4\log{A})$ for the lattice reduction costs using classical arithmetic.  One must also add the cost of multi-factor Hensel lifting which is $\mathcal{O}(N^6+N^4 \log^2{A})$ ignoring the small terms $\log(r)$ and $\log^2{p}$ (see~\cite{mca} for details).  This is the first improvement over the Sch\"onhage bound given in 1984~\cite{Schonhage} of $\mathcal{O}(N^8+N^5\log^3{A})$.  

The second new complexity result comes in the problem of reconstructing a minimal polynomial from a complex approximation of the algebraic number.  In this application we know $\mathcal{O}(d^2+d\log{H})$ bits of an approximation of some complex root of an unknown polynomial $h(x)$ with degree $d$ and with maximal coefficient of absolute value $\leq H$.  Then our algorithm can be used to find the coefficients of $h(x)$ in $\mathcal{O}(d^7 + d^5 \log^2{H})$ CPU operations.

Other problems of common interest which might be impacted by our algorithm include integer relation finding (where $N = 1$) and simultaneous Diophantine approximation of several real numbers \cite{hanrot,bright} (where $r = 1$).  

\noindent \textbf{Notations:} All costs are given for the bit-complexity model.  A standard row vector will be denoted ${\bf v}$, ${\bf v}[i]$ represents the $i^{\textrm{th}}$ entry of ${\bf v}$, ${\bf v}[i, \ldots, j]$ a vector consisting of all entries of ${\bf v}$ from the $i^{\textrm{th}}$ entry to the $j^{\textrm{th}}$ entry, and ${\bf v}[-1]$ the final entry of ${\bf v}$.  Also we will use $\| {\bf w} \|_{\infty}$ as the max-norm or the largest absolute value of an entry in the vector ${\bf w}$, $\parallel {\bf w} \parallel :=\sqrt{ \sum ({\bf w}[i])^2}$ which we call the norm of ${\bf w}$, and ${\bf w}^{T}$ as the transpose of ${\bf w}$.  The scalar product will be denoted ${\bf v} \cdot {\bf w}:= \sum {\bf v}[i] \cdot {\bf w}[i]$.  For a matrix $M$ we will use $M[1,\ldots, k]$ to denote the first $k$ columns of $M$.  The $n$ by $n$ identity matrix will be denoted $I_{n \times n}$.  For a real number $x$ we use $\lceil x \rceil$ and $\lfloor x \rfloor$ to denote the closest integer $\geq x$ and $\leq x$ respectively.

\noindent \textbf{Road map:} In section~\ref{background} we give a brief introduction to lattice reduction algorithms.  In section~\ref{mainalgsec} we present the central algorithm of the paper and prove its correctness.  In section~\ref{adsect} we prove several important features by studying quasi-invariants we call the active determinant and progress.  In this section we treat lattice reduction as a black-box algorithm.  In section~\ref{compsect} we prove the overall complexity and other important claims about the new algorithm by fixing a choice for a standard lattice reduction algorithm.  In section~\ref{newcomplexities} we offer new complexity results for factoring polynomials in $\mathbb{Z}[x]$ and algebraic number reconstruction.

\section{Background on lattice reduction}\label{background}

\noindent The purpose of this section is to present some facts from~\cite{LLL} that
will be needed throughout the paper.  For a more general treatment of lattice reduction see~\cite{lovasz}.

A lattice, $L$, is a discrete subset of $\mathbb{R}^n$ that is also a $\mathbb{Z}$-module.
Let ${\bf b}_1,\ldots,{\bf b}_d \in L$ be a basis of $L$ and denote ${{\bf b}_1^*},\ldots,{{\bf b}_d^*} \in \mathbb{R}^n$ as the 
Gram-Schmidt orthogonalization over $\mathbb{R}$ of ${{\bf b}_ 1},\ldots,{{\bf b}_ d}$.  Let $\delta \in (1/4, 1]$ and $\eta \in [1/2, \sqrt{\delta})$.  
Let $l_i = \log_{1/\delta}{\parallel {\bf b}_{i}^* \parallel}^2$,
 and denote $\mu_{i,j}=\frac{{\bf b}_{i} \cdot {\bf b}_ j^*}{{\bf b}_ j^* 
\cdot {\bf b}_ j^*}$. Note that ${\bf b}_{i}, {\bf b}_{i}^*, l_i, \mu_{i,j}$ will change throughout the algorithm sketched below.

\begin{definition}\label{reducedbasis}
${{\bf b}_ 1}, \ldots, {{\bf b}_ d}$ is \emph{LLL-reduced}
if ${\parallel {\bf b}_{i}^* \parallel}^2 \leq \frac{1}{\delta-\mu_{i+1,i}^2}{\parallel {{\bf b}_{i+1}^*} \parallel}^2$ for $1 \leq i < d$ and $|\mu_{i,j}| \leq \eta$ for $1 \leq j < i \leq d$.
\end{definition}

In the original paper the values for $(\delta, \eta)$ were chosen as $(3/4, 1/2)$ so that $\frac{1}{\delta-\eta^2}$ would simply be 2.

\begin{alg}[Rough sketch of LLL-type algorithms]\label{LLL} \mbox{} \\
\noindent \emph{Input:} A basis ${{\bf b}_ 1},\ldots,{{\bf b}_ d}$ of a lattice $L$. \\
\noindent \emph{Output:} An LLL-reduced basis of $L$.
\begin{enumerate}[A -]
\item $\kappa:= 2$

\item \textbf{while} $\kappa \leq d$ \textbf{do:}
\begin{enumerate}[1 -]
\item{\em{(Gram-Schmidt over $\mathbb{Z}$)}}.\label{subtract} By subtracting suitable $\mathbb{Z}$-linear 
combinations of ${{\bf b}_ 1},\ldots,{{\bf b}_ {\kappa-1}}$ from ${{\bf b}_{\kappa}}$ make sure that $| \mu_{i,\kappa} | \leq \eta$ for $i<\kappa$.

\item{\em{(LLL Switch)}}.\label{swap} If interchanging ${{\bf b}_{\kappa-1}}$ and ${{\bf b}_{\kappa}}$ will 
decrease $l_{\kappa-1}$ by at least 1 then do so.

\item{\em{(Repeat)}}.  If not switched $\kappa:=\kappa+1$, if switched $\kappa = \textrm{max}(\kappa-1,2)$.

\end{enumerate}

\end{enumerate}
\end{alg}
That the above algorithm terminates, and that the output is
LLL-reduced was shown in~\cite{LLL}.  Step \ref{subtract} has no effect on the $l_i$.  In step~\ref{swap} the only $l_i$ that change
are $l_{\kappa-1}$ and $l_{\kappa}$.  The following lemmas present some standard facts which we will need.

\begin{lemma}\label{maxlll}  An LLL switch can not increase $\max(l_1, \ldots, l_d)$,
nor can it decrease $\min(l_1, \ldots, l_d)$.
\end{lemma}

\begin{lemma}\label{Bbound} If ${\parallel {\bf b}_d^* \parallel} > B$ then any vector in $L$ with norm 
$\leq B$ is
a $\mathbb{Z}$-linear combination of ${{\bf b}_ 1}, \ldots, {{\bf b}_ {d-1}}$.
\end{lemma}
In other words, if the current basis of the lattice is ${{\bf b}_ 1},\ldots,{{\bf b}_ d}$ and if the last vector
has sufficiently large G-S length then, provided the user is only interested in elements of $L$ with norm $\leq B$, the last basis element can be removed.

Lemma~\ref{Bbound} follows from the proof of~\cite[Eq.\ (1.11)]{LLL}, and is true
regardless of whether ${{\bf b}_ 1},\ldots,{{\bf b}_ d}$ is LLL-reduced or not.
However, if one chooses an arbitrary basis ${{\bf b}_ 1},\ldots,{{\bf b}_ d}$ of some lattice $L$, then it is unlikely that the last vector has large G-S length (after all, $\parallel {{\bf b}_ d}^* \parallel$
is the norm of ${{\bf b}_ d}$ reduced modulo ${{\bf b}_ 1},\ldots,{{\bf b}_ {d-1}}$ over $\mathbb{R}$).
The effect of LLL reduction is to move G-S length towards later
vectors.

\section{Main algorithm}\label{mainalgsec}

In this section we present the central algorithm of the paper and a proof of its correctness.  Our algorithm is a kind of wrapper for other standard lattice reduction algorithms.  We try to present it as independently as possible of the choice of lattice reduction algorithm.  In order to be general we must first outline the features that we require of the chosen lattice reduction algorithm.  Our first requirement is that the output satisfy the following slightly weakened version of LLL-reduction.

\begin{definition}\label{alphareduced}
Let $L \subseteq \mathbb{R}^n$ be a lattice and ${{\bf b}_ 1},\ldots, {{\bf b}_ s} \in L$ be $\mathbb{R}$-linearly independent.  We call ${{\bf b}_ 1}, \ldots, {{\bf b}_ s}$ an $\alpha$-reduced basis of $L$ if \ref{one},\ref{two}, and \ref{threea} hold, and an $(\alpha,B)$-reduced sequence (basis of a sub-lattice) if \ref{one},\ref{two}, and \ref{threeb} hold:

\begin{enumerate}
\item \label{one} $\parallel {\bf b}_{i}^* \parallel \leq \alpha \parallel {\bf b}_{i+1}^* \parallel$ for $i=1 \ldots s-1$.

\item \label{two} $\parallel {\bf b}_{i}^* \parallel \leq \parallel {\bf b}_{i} \parallel \leq \alpha^{i-1} \parallel {\bf b}_{i}^* \parallel$ for $i=1 \ldots s$.

\item \begin{enumerate}

\item \label{threea} $L=\mathbb{Z}{{\bf b}_ 1} +\cdots + \mathbb{Z}{{\bf b}_ s}$.

\item \label{threeb} $\parallel {\bf b}_s^* \parallel \leq B$ and for every ${\bf v} \in L$ with $\parallel {\bf v} \parallel \leq B$ we have ${\bf v} \in \mathbb{Z}{{\bf b}_ 1} + \cdots + \mathbb{Z}{{\bf b}_ s}$.
\end{enumerate}

\end{enumerate}

\end{definition}

The original LLL algorithm from~\cite{LLL} returns output with $\alpha=\sqrt{2}$, ${\rm L}^2$ from~\cite{fpLLL} with $\alpha=\sqrt{ \frac{1}{\delta-\eta^2}}$ for appropriate choices of $(\delta, \eta)$, and H-LLL from~\cite{HLLL} reduced with $\alpha= \frac{ \theta \eta + \sqrt{(1+\theta^2)\delta-\eta^2}}{\delta - \eta^2}$ for appropriate $(\delta, \eta, \theta)$.  We may now also make a useful observation about an $(\alpha,B)$-reduced sequence.

\begin{lemma}\label{smallbreduced}
If the vectors ${{\bf b}_ 1}, \ldots, {{\bf b}_ s}$ form an $(\alpha,B)$-reduced sequence and we let ${\bf b}_1^*, \ldots, {\bf b}_s^*$ represent the GSO, then the following properties are true:

\begin{itemize}
\item $\parallel {\bf b}_{i}^* \parallel \leq \alpha^{s-i} B$ for all $i$.

\item $\parallel {\bf b}_{i} \parallel \leq \alpha^{s-1} B$ for all $i$.

\end{itemize}

\end{lemma}

We use the concept of $\alpha$-reduction as a means of making proofs which are largely independent of which lattice reduction algorithm a user might choose.  For a basis which is $\alpha$-reduced, a small value of $\alpha$ implies a strong reduction.  In our algorithm we use the variable $\alpha$ as the worst-case guarantee of reduction quality.  We make our proofs (specifically Lemma~\ref{ADincrease} and Theorem~\ref{numloops}) assuming an $\alpha \geq \sqrt{4/3}$.  This value is chosen because~\cite{LLL,fpLLL,HLLL} cannot guarantee a stronger reduction.  An $(\alpha,B)$-reduced bases is typically made from an $\alpha$-reduced basis by removing trailing vectors with large G-S length.  The introduction of $(\alpha,B)$-reduction does not require creating new lattice reduction algorithms, just the minor adjustment of detecting and removing vectors above a given G-S length.

\begin{alg}\label{lllwithremovals}

\textbf{LLL\_with\_removals}

\textbf{Input:} ${{\bf b}_ 1}, \ldots, {{\bf b}_ s} \in \mathbb{R}^n $ and $B \in \mathbb{R}$.

\textbf{Output:} ${{\bf b'}_1}, \ldots, {{\bf b'}_{s'}} \in \mathbb{R}^n$ $(\alpha,B)$-reduced, $s' \leq s$.

\textbf{Procedure:}  Use any lattice reduction procedure which returns an $\alpha$-reduced basis and follows Assumption~\ref{switchreqs}.  However, when it is discovered that the final vector has G-S length provably $>B$ remove that final vector (deal with it no further).
\end{alg}

\begin{ass}\label{switchreqs}
The lattice reduction algorithm chosen for LLL\_with\_removals must use switches of consecutive vectors during its reduction process.  These switches must have the following properties:

\begin{enumerate}
\item There exists a number $\gamma > 1$ such that every switch of vectors ${\bf b}_{i}$ and ${{\bf b}_ {i+1}}$ increases $\parallel {\bf b}_{i+1}^* \parallel^2$ by a factor provably $\geq \gamma$.

\item The quantity $\textrm{max}\{ \parallel {\bf b}_{i}^* \parallel, \parallel {\bf b}_{i+1}^* \parallel \}$ cannot be increased by switching ${\bf b}_{i}$ and ${{\bf b}_ {i+1}}$.

\item No steps other than switches can affect G-S norms $\parallel {\bf b}_1^* \parallel, \ldots, \parallel {\bf b}_s^* \parallel $.
\end{enumerate}
\end{ass}

Assumption~\ref{switchreqs} is not very strong as~\cite{LLL,fpLLL,HLLL,schnorr2,arne} and the sketch in Algorithm~\ref{LLL} all conform to these assumptions.  We do not allow for the extreme case where $\gamma=1$, although running times have been studied in~\cite{akhavi,flags}.  It should also be noted that in the floating point lattice reduction algorithms $\parallel {\bf b}_s ^* \parallel$ is only known approximately.  In this case one must only remove vectors whose approximate G-S length is sufficiently large to ensure that the exact G-S length is $\geq B$.

The format of the input matrices was given in section~\ref{intro}.  A search parameter $B$ is given to bound the norm of the target vectors.  The algorithm performs its best when $B$ is small compared to the bit-length of the entries in the input matrix, although $B$ need not be small for the algorithm to work.

\begin{definition}\label{pj_def} We say the $P_j$ are {\em large enough} if:
\begin{equation}\label{pj}
|P_j| \geq 2\alpha^{4r+4k+2}B^2 \textrm{~for all but $k=\mathcal{O}(r)$ values of $j$}.
\end{equation}
\end{definition}

\noindent Note that if $N=\mathcal{O}(r)$ then the $P_j$ are trivially large enough.  However, for applications where $N$ is potentially much larger than $r$ this becomes a non-trivial condition.  In this case having $B$ close to $X$ means that the $P_j$'s are not large enough.

\medskip

In the following algorithm we will gradually reduce the input basis.  This will be done one column at a time, similar to the experiments in~\cite{Belabas,bright}.  The current basis vectors are denoted ${\bf b}_{i}$ and we will use $M$ to represent the matrix whose rows are the ${\bf b}_{i}$.  We will use the notation ${{\bf x}_j}$ to represent the column vector $(x_{1,j}, \ldots, x_{r,j})^{T}$.

The matrix $M$ will begin as $I_{r \times r}$, and we will adjoin ${\bf x}_1$ and a new row $({\bf 0}, P_1)$ if appropriate.  Each time we add a column ${\bf x}_j$ we will need to calculate the effects of prior lattice reductions on the new ${\bf x}_j$.  We use ${\bf y}_j$ to represent a new column of entries which will be adjoined to $M$.  In fact ${\bf y}_j = M[1,\ldots, r] \cdot {\bf x}_j$.  Before adjoining the entries we also scale them by a power of 2, to have smaller absolute values.  This keeps the entries in $M$ at a uniform absolute value.  The central loop of the algorithm is the process of gradually using more and more bits of ${\bf y}_j$ until every entry in $M$ is again an integer.  No rounding is performed: we use rational arithmetic on the last column of each row.  Throughout the algorithm the number of rows of $M$ will be changing.  We let $s$ be the current number of rows of $M$.  If~\eqref{pj} is satisfied for some $k=\mathcal{O}(r)$ then we can actually bound $s$ by $2r+2k+1$.  We use $c$ as an apriori upper bound on $s$, either $c:=2r+2k+1$ or $c:=r+N$.  The algorithm has better performance when $c$ is small.  We let $L$ represent the lattice generated by the rows of $A$.

\begin{alg}\label{mainalg} 
{\rm \textbf{Gradual\_LLL}}

\textbf{Input:} A search parameter, $B \geq \sqrt{5} \in \mathbb{Q}$, an integer knapsack-type matrix, $A$, and an $\alpha \geq \sqrt{4/3}$.

\textbf{Output:} An $(\alpha,B)$-reduced basis ${{\bf b}_ 1}, \ldots, {{\bf b}_ s}$ of a sub-lattice $L'$ in $L$ with the property that if ${\bf v} \in L$ and ${\parallel {\bf v} \parallel} \leq B$ then ${\bf v} \in L'$.
\end{alg}

\textbf{The Main Algorithm:}

\begin{enumerate}[1 -]
\item\label{setc} \textbf{if}~\eqref{pj} holds set $c:=\textrm{min}(2r+2k+1,r+N)$ 

\item\label{init} $s:=r; M := I_{r \times r}$

\item\label{clmn} \textbf{for } $j = 1 \ldots N$ \textbf{ do}: 
\begin{enumerate}[a -]

\item\label{newent} ${\bf y}_j:= M[1,\ldots,r] \cdot {\bf x}_j$;  $\ell :=\lfloor \log_2{(\textrm{max}\{|P_j|, \|{\bf y}_j\|_{\infty}, 2 \})} \rfloor$

\item\label{adjoin} $M := \left[ \begin{array}{c|c} 0 & P_j/2^{\ell} \\ \hline M & {\bf y}_j/{2^{\ell}} \\ \end{array} \right]$; \textbf{if }$P_j \neq 0$\textbf{ then }$s := s+1$ \textbf{else} remove zero row 

\item\label{mainloop}\textbf{while } $(\ell \neq 0)$ \textbf{ do:}

\begin{enumerate}[i -]

\item\label{newprec} ${\bf y}_j := 2^{\ell} \cdot M \cdot [ 0,  \cdots, 0,1]^{T}$; $\ell := \textrm{max} \{ 0, \lceil \log_2 {( \frac{\|{\bf y}_j\|_{\infty}}{{\alpha}^{2c}B^2} )} \rceil \}$

\item\label{scale} $M:= \left[ \begin{array}{c|c} M[1, \ldots, r+j-1] &  {{\bf y}_j}/{2^{\ell}} \end{array} \right]$

\item\label{singlelll} Call $\textrm{LLL\_with\_removals}$ on $M$ and set $M$ to output; adjust $s$

\end{enumerate}

\end{enumerate}

\item \textbf{return} $M$

\end{enumerate}

First we will prove the correctness of the algorithm.  We need to show that the Gram-Schmidt lengths are never decreased by scaling the final entry or adding a new entry.

\begin{lemma}\label{scalenondecrease}
Let ${\bf b}_1, \ldots, {\bf b}_s \in \mathbb{R}^n$ be the basis of a lattice and ${\bf b}_1^*, \ldots, {\bf b}_s^*$ its GSO.  Let $\sigma:\mathbb{R}^{n} \to \mathbb{R}^{n}$ scale up the last entry by some 
factor $\beta >1$, then we have $\parallel {\bf b}_i^* \parallel \leq \parallel {\sigma({\bf b}_i)}^* \parallel$.  In other words, scaling the final entry of each vector by the same scalar $\beta > 1$ cannot decrease $\|{\bf b}_i^*\|$ for any $i$.  
\end{lemma}

\begin{lemma}\label{extraentry}
Let ${\bf b}_1, \ldots, {\bf b}_s \in \mathbb{R}^{n}$ and let ${\bf b}_1^*, \ldots, {\bf b}_s^* \in \mathbb{R}^{n}$ be their GSO.  The act of adjoining an ${(n+1)}^{\textrm{st}}$ entry to each vector and re-evaluating the GSO cannot decrease $\parallel {\bf b}_i^* \parallel$ for any $i$ (assuming that the new entry is in $\mathbb{R}$). 
\end{lemma}

The proofs of these lemmas are quite similar and can be found in the appendix.  Now we are ready to prove the first theorem, asserting the correctness of algorithm~\ref{mainalg}'s output.

\begin{theorem}\label{taccuracy}
Algorithm~\ref{mainalg} correctly returns an $\alpha$-reduced basis of a sub-lattice, $L'$, in $L$ such that if ${\bf v} \in L$ and $\parallel {\bf v} \parallel \leq B$ then ${\bf v} \in L'$.
\end{theorem}

\begin{proof}
When the algorithm terminates all entries are unscaled and each vector in the output is inside of $L$ as it is a linear combination of the original input vectors.  Thus the output is a basis of a sub-lattice $L'$ inside $L$.  Further, the algorithm terminates after a final call to step~\ref{singlelll} so returns an $(\alpha,B)$-reduced sequence.

Now we show that if ${\bf v} \in L$ and $\parallel {\bf v} \parallel \leq B$ then ${\bf v} \in L'$.  The removed vectors correspond to vectors $\tilde{ {\bf b}_i } \in L$ that, by lemmas~\ref{scalenondecrease} and~\ref{extraentry}, have G-S length at least as large as those of ${\bf b}_i$.  The claim then follows from lemmas~\ref{maxlll} and~\ref{Bbound}.
\end{proof}

\section{Two invariants of the algorithm}

Here we present the important proofs about the set-up of our algorithm.  All proofs in this section and the next allow for a black-box lattice reduction algorithm up to satisfying assumption~\ref{switchreqs}.  Each proof in this section involves the study of an invariant.  The two invariants which we use are:
\begin{itemize}
\item The Active Determinant, $\textrm{AD}(M)$, which is the product of the G-S lengths of the active vectors.  This remains constant under standard lattice reduction algorithms, and allows us to bound many features of the proofs.

\item The Progress, $PF= \sum_{i=1}^s (i-1) \log\parallel {\bf b}_{i}^* \parallel^2 + n_{\textrm{rm}} r \log(4\alpha^{4c}B^4)$, where $n_{\textrm{rm}}$ is the total number of vectors which have been removed so far. This function is an energy function which never decreases, and is increased by $\geq  1$ for each switch made in the lattice reduction algorithm.
\end{itemize}

\subsubsection*{A study of the active determinant}\label{adsect}

\begin{definition}
We call the active determinant of the vectors ${{\bf b}_ 1}, \ldots, {{\bf b}_ s}$ the product of their Gram-Schmidt lengths.  For notation we use, AD or $\textrm{AD}(\{ {\bf b}_{i} \} ):= \prod_{i=1}^s {\parallel {\bf b}_{i}^* \parallel}$.  For a matrix $M$ with the $i^{\textrm{th}}$ row denoted by $M[i]$, we use $\textrm{AD}$ or $\textrm{AD}(M)=\textrm{AD}(\{M[1], \ldots, M[s]\})$.
\end{definition}

For an $(\alpha,B)$-reduced sequence we can nicely bound the AD.  We have such a sequence after each execution of step~\ref{singlelll}.

\begin{lemma}\label{ADreduced}
If ${{\bf b}_ 1}, \ldots, {{\bf b}_ s}$ are an $(\alpha,B)$-reduced sequence then $\textrm{AD} \leq {({\alpha}^{s-1}B^2)}^{s/2}$.
\end{lemma}

We now want to attack two problems, bounding the norm of each vector just before lattice reduction, and bounding the number of vectors throughout the algorithm.

\begin{lemma}\label{smallvectors}
If $s \leq c$ then just before step~\ref{singlelll} we have $\parallel {\bf b}_i \parallel^2 \leq 2\alpha^{4c}B^4$ for $i=1 \ldots s$.
\end{lemma}

The full details of this proof can be found in the appendix.  The following theorem holds trivially when there is no condition on the $P_j$ or if $N=0$.  When $N > r$ and $B$ is at least a bit smaller than $X$ we can show that not all of the extra vectors stay in the lattice.  In other words, if there is enough of a difference between $B$ and $X$ then the sub-lattice aspect of the algorithm begins to allow for some slight additional savings.  Here the primary result of this theorem is allowing $\mathcal{O}(r)$ vectors with a relatively weak condition on the $P_j$.

\begin{theorem}\label{activevectors}
Throughout the algorithm we have $s \leq c$. 
\end{theorem}

\begin{proof}
If $c=r+N$ then $s \leq c$ is vacuously true.  So assume $c= 2(r+k)+1$ and all but $k=\mathcal{O}(r)$ of the $P_j$ satisfy $|P_j| \geq 2\alpha^{4r+4k+2}B^2$.  When the algorithm begins, $\textrm{AD}=1$ and $s=r$.  For $s$ to increase step~\ref{clmn} must finish without removing a vector.  If this happens during iteration $j$ then the $\textrm{AD}$ has increased by a factor $|P_j|$.  The LLL-switches inside of step~\ref{singlelll} do not alter the AD by Assumption~\ref{switchreqs}.  Each vector which is removed during step~\ref{singlelll} has G-S length $\leq 2\alpha^{4r+4k+2}B^2$ by Lemmas~\ref{smallvectors} and~\ref{maxlll}.   After iteration $j$ we have $n_{\textrm{rm}}=r+j-s$ as the total number of removed vectors.  All but $k$ of the $P_i$ have larger norm than any removed vector. Therefore the smallest $AD$ can be after iteration $j$ is $\geq {(2\alpha^{(4r+4k+2)}B^2)}^{j-k-n_{\textrm{rm}}}$.  Rearranging we get $\textrm{AD} \geq {(2\alpha^{4r+4k+2}B^2)}^{s-r-k}$.  This contradicts Lemma~\ref{ADreduced} when $s$ reaches $2r+2k$ for the first time because $(2\alpha^{4r+4k+2}B^2)^{r+k} \geq {(\alpha^{2r+2k-1}B^2)}^{r+k}$.  \end{proof}

\begin{corollary}\label{smallvectors2}
Throughout the algorithm we have $\parallel {\bf b}_{i}^* \parallel \leq 2\alpha^{2c}B^2$.
\end{corollary}

We also use the active determinant to bound the number of iterations of the main loop, i.e. step~\ref{mainloop}.  First we show in the appendix that AD is increased by every scaling which does not end the main loop.

\begin{lemma}\label{ADincrease}
Every execution of step~\ref{scale} either increases the $\textrm{AD}$ by a factor $\geq \frac{{\alpha}^{c}B}{2}$ or sets $\ell=0$.
\end{lemma}

Now we are ready to prove that the number of iterations of the main loop is $\mathcal{O}(r+N)$.  This is important because it means that, although we look at all of the information in the lattice, the number of times we have to call lattice reduction is unrelated to $\log{X}$.

\begin{theorem}\label{numloops}
The number of iterations of step~\ref{mainloop} is $\mathcal{O}(r+N)$.
\end{theorem}

The strategy of this proof is to show that each succesful scaling increases the active determinant and to bound the number of iterations using Lemma~\ref{ADreduced} and Corollary~\ref{smallvectors2}.  For space constraints this proof is provided in the appendix.

\subsubsection*{A study of the progress function}\label{psect}

We will now amortize the costs of lattice reduction over each of the $\mathcal{O}(r+N)$ calls to step~\ref{singlelll}.  We do this by counting switches, using Progress $PF$ (defined below).  In order to mimic the proof from~\cite{LLL} for our algorithm we introduce a type of Energy function which we can use over many calls to LLL (not only a single call).

\begin{definition}
Let ${{\bf b}_ 1}, \ldots, {{\bf b}_ s}$ be the current basis at any point in our algorithm, let ${\bf b}_1^*, \ldots, {\bf b}_ s^*$ be their GSO, and $l_i := \log_{\gamma}{\parallel {\bf b}_{i}^* \parallel^2}$ for all $i=1\ldots s$.  We let $n_{\textrm{rm}}$ be the number of vectors which have been removed so far in the algorithm.  Then we define the progress function $PF$ to be:

$$PF:= 0 \cdot l_0 + \cdots + (s-1)\cdot l_s + n_{\textrm{rm}}\cdot c \cdot \log_{\gamma}{(4\alpha^{4c}B^4)}.$$
\end{definition}

This function is designed to effectively bound the largest number of switches which can have occurred so far.  To prove that it serves this purpose we must prove the following lemma:

\begin{lemma}\label{Pincreases}
After step~\ref{init} Progress $PF$ has value 0.  No step in our algorithm can cause the progress $PF$ to decrease.  Further, every switch which takes place in step~\ref{singlelll} must increase $PF$ by at least 1.
\end{lemma}

\begin{theorem}\label{switchcomplexity}
Throughout our algorithm the total number of switches used by all calls to step~\ref{singlelll} is $\mathcal{O}((r+N)c(c+\log{B}))$ with $P_j$ and  $\mathcal{O}(c^2(c+\log{B}))$ with no $P_j$.
\end{theorem}

\begin{proof}
Since Lemma~\ref{Pincreases} shows us that $PF$ never decreases and every switch increases $PF$ by at least 1, then the number of switches is bounded by $PF$.  However $PF$ is bounded by Lemma~\ref{smallvectors} which bounds $l_i \leq \log_{\gamma}{(\alpha^{4r}B^4)}$, Theorem~\ref{activevectors} which bounds $s \leq c$, and the fact that we cannot remove more vectors than are given which implies $n_{\textrm{rm}} \leq r+N$.  Further we can see that $(s-1)l_s \leq (c-1)\log_{\gamma}{(4\alpha^{4c}B^4)}$ so $PF$ is maximized by making $n_{\textrm{rm}} = (r+N)$ (or $c$ if no vectors added) and $s=0$.  In which case we have $\textrm{number of switches} \leq PF \leq (r+N)(c-1)(\log_{\gamma}{(4\alpha^{4c}B^4)} = \mathcal{O}((r+N)c(c+\log{B}))$.  Also if there are no $P_j$, we can replace $r+N$ by $c$. 
\end{proof}

\section{Complexity bound of main algorithm}\label{compsect}

In this section we wish to prove a bound for the overall bit-complexity of algorithm~\ref{mainalg}.  The complexity bound must rely on the complexity bound of the lattice reduction algorithm we choose for step~\ref{singlelll}.  The results in the previous sections have not relied on this choice.  We will present our complexity bound using the H-LLL algorithm from~\cite{HLLL}.  We choose H-LLL for this result because of its favorable complexity bound and because the analysis of our necessary adaptations is relatively simple.  See~\cite{HLLL} for more details on H-LLL.

We make some minor adjustments to the H-LLL algorithm and its analysis.  The changes to the algorithm are the following: \begin{itemize}
\item We have a single non-integer entry in each vector of bit-length $\mathcal{O}(c+\log{X})$.

\item Whenever the final vector has G-S length sufficiently larger than $B$, it is removed.  This has no impact on the complexity analysis.
\end{itemize}

We use $\tau$ as the number of switches used in a single call to H-LLL.  This allows the analysis of progress $PF$ to be applied directly.  The following theorem is an adaptation of the main theorem in~\cite{HLLL} adapted to reflect our adjustments.  

\begin{theorem}\label{taucost}
If a single call to step~\ref{singlelll}, with H-LLL~\cite{HLLL} as the chosen variation of LLL, uses $\tau$ switches then the CPU cost is bounded by $\mathcal{O}((\tau+c+\log{B})c^2[(r+N)(c+\log{B})+\log{X}])$ bit-operations.
\end{theorem}

Now we are ready to complete the complexity analysis of the our algorithm.

\begin{theorem}\label{overallcomplexity}
The cost of executing algorithm~\ref{mainalg} with H-LLL~\cite{HLLL} as the variant of LLL in step~\ref{singlelll} is $$ \mathcal{O}((r+N)c^3(c+\log{B})[\log{X}+(r+N)(c+\log{B})])$$ CPU operations, where $B$ is a search parameter chosen by the user, $|A[i,j]| \leq X$ for all $i,j$, and $c=r+N$ or $c=\mathcal{O}(r)$ (see definition~\ref{pj_def} for details). If there are no $P_j$'s then the cost is $$\mathcal{O}((r+N+c^2)(c+\log{B})c^2[\log{X}+(r+N)(c+\log{B})]).$$
\end{theorem}

\begin{proof}
Steps~\ref{init},~\ref{adjoin},~\ref{newprec}, and~\ref{scale} have negligible costs in comparison to the rest of the algorithm.  Step~\ref{newent} is called $N$ times, each call performs $s$ inner products.  While each inner product performs $r$ multiplications each of the form ${\bf b}_{i}[m] \cdot x_{m,j}$ appealing to Corollary~\ref{smallvectors2} we bound the cost of each multiplication by $\mathcal{O}( (c+\log{B}) \log{X})$.  Since Theorem~\ref{activevectors} gives $s \leq c$ we know that the total cost of all calls to step~\ref{newent} is $\mathcal{O}(Ncr(c+\log{B})\log{X})$.  Let $k =\mathcal{O}(r+N)$ be the number of iterations of the main loop.  Let $\tau_i$ be the number of LLL switches used in the $i^{\textrm{th}}$ iteration.  Theorem~\ref{taucost} gives the cost of the $i^{\textrm{th}}$ call to step~\ref{singlelll} as $= \mathcal{O}((\tau_i+c+\log{B})c^2 [(r+N)(c+\log{B})+\log{X}])$.  Theorem~\ref{switchcomplexity} implies that $\tau_1 + \cdots + \tau_k =\mathcal{O}((r+N)c(c+\log{B}))$ (or $\mathcal{O}(c^2(c+\log{B}))$ when there are no $P_j$'s).  The total cost of all calls to step~\ref{singlelll} is then $\mathcal{O} ([k(c+\log{B})+\tau_1 + \cdots + \tau_k ]c^2 [(r+N)(c+\log{B})+\log{X}])$.  The term $[k(c+\log{B})+\tau_1 + \cdots + \tau_k]$ can be replaced by $\mathcal{O}((r+N)c(c+\log{B}))$ (if no $P_j$ then $\mathcal{O}((r+N+c^2)(c+\log{B}))$).  The complete cost of is now $\mathcal{O}(Nrc(c+\log{B})\log{X} + (r+N)c^3[c+\log{B})(\log{X}+(r+N)(c+\log{B})])$.  The first term is absorbed by the cost of the second term, proving the theorem.  If there are no $P_j$ then we get $\mathcal{O}((r+N+c^2)(c+\log{B})c^2[\log{X}+(r+N)(c+\log{B})])$.
\end{proof}

\section{New complexities for applications of main algorithm}\label{newcomplexities}

Our algorithm has been designed for some applications of lattice reduction.  In this section we justify the importance of this algorithm by directly applying it to two classical applications of lattice reduction.

\subsubsection*{New complexity bound for factoring in $\mathbb{Z}[x]$}\label{polynomials}

In~\cite{BHKS} it is shown that the problem of factoring a polynomial, $f \in \mathbb{Z}[x]$, can be accomplished by the reduction of a large knapsack-type lattice.  In this subsection we merely apply our algorithm to the lattice suggested in ~\cite{BHKS}.

\vspace{.1in}

\noindent \textbf{Reminders from~\cite{BHKS}.} Let $f \in \mathbb{Z}[x]$ be a polynomial of degree $N$.  Let $A$ be a bound on the absolute value of the coefficients of $f$.  Let $p$ be a prime such that $f \equiv 
l_f f_1 \cdots f_r \/~\mathrm{mod}\/~p^a$ a separable irreducible factorization of $f$ in the $p$-adics lifted to precision $a$, the $f_i$ are monic, and $l_f$ is the leading coefficient of $f$.  For our purposes we choose $B:=\sqrt{r+1}$.  

We will make some minor changes to the All-Coefficients matrix defined in~\cite{BHKS} 
to produce a matrix that looks like:

$$\left( \begin{array}{cccccc}
 &  &  &  &  & p^{a-b_N} \\ 
 &  &  &  & \Ddots &  \\ 
 &  &  & p^{a-b_1} &  &  \\ 
1 &  &  & x_{1,1} & \cdots & x_{1,N} \\ 
 & \ddots &  & \vdots & \ddots & \vdots \\ 
 &  & 1 & x_{r,1} & \cdots & x_{r,N}
\end{array} \right).$$

Here $x_{i,j}$ is the $j^{\textrm{th}}$ coefficient of $f'_i \cdot f / f_i \textrm{ mods }p^a$ divided by $p^{b_j}$ and $p^{b_j}$ represents $\sqrt{N}$ times a bound on the $j^{\mathrm{th}}$ coefficient of $g' \cdot f/g$ for any true factor $g \in \mathbb{Z}[x]$ of $f$.  In this way the target vectors will be quite small.  An empty spot in this matrix represents a zero entry.  This matrix has $p^{a-b_j} > 2^{N^2+N\log(A)} > 2{\alpha}^{4r+2}B^2$ for all $j$.  An $(\alpha,B)$-reduction of this matrix will solve the recombination problem by a similar argument to the one presented in~\cite{BHKS} and refined in~\cite{phd}.  Now we look at the computational complexity of making and reducing this matrix which gives the new result for factoring inside $\mathbb{Z}[x]$.

\begin{theorem}\label{polytheorem}
Using algorithm~\ref{mainalg} on the All-Coefficients matrix above provides a complete irreducible factorization of a polynomial $f$ of degree $N$, coefficients of bit-length $\leq \log{A}$, and $r$ irreducible factors when reduced modulo a prime $p$ in $$\mathcal{O}(N^2r^4[N+\log{A}])$$ CPU operations.  The cost of creating the All-Coefficients matrix adds $\mathcal{O}(N^4[N^2+\log^2{A}])$ CPU operations using classical arithmetic (suppressing small factors $\log{r}$ and $\log^2{p}$) to the complexity bound.
\end{theorem}

The following chart gives a complexity bound comparison of our algorithm with the factorization algorithm presented by Sch\"onhage in~\cite{Schonhage} we estimate both bounds using classical arithmetic and fast FFT-based arithmetic~\cite{fft}.  We also suppress all $\log{N}$, $\log{r}$, $\log{p}$, and $\log \log {A}$ terms.

\begin{center} 
\begin{tabular}{|c|c|}
\hline
Classical Gradual\_LLL & \raisebox{11pt}{~}$\mathcal{O}(N^3r^4 + N^2r^4\log{A} + N^6 + N^4\log^2{A})$\\
\hline
Classical Sch\"onhage & \raisebox{11pt}{~}$\mathcal{O}(N^8+N^5\log^3{A})$\\
\hline
Fast Gradual\_LLL & \raisebox{11pt}{~}$\mathcal{O}(N^3r^3+N^2r^3\log{A})$\\
\hline
Fast Sch\"onhage & \raisebox{11pt}{~}$\mathcal{O}(N^6+N^4\log^2{A})$\\
\hline
\end{tabular}
\end{center}

The Sch\"onhage algorithm is not widely implemented because of its impracticality.  For most polynomials, $r$ is much smaller than $N$.  Our main algorithm will reduce the All-Coefficients matrix with a competitive practical running time, but constructing the matrix itself will require more Hensel lifting than seems necessary in practice.  In~\cite{phd} a similar switch-complexity bound to section~\ref{psect} is given on a more practical factoring algorithm.

\subsubsection*{Algebraic number reconstruction}\label{lenstra}

The problem of finding a minimal polynomial from an approximation of a complex root was attacked in~\cite{lenstra} using lattice reduction techniques using knapsack-type bases.  For an extensive treatment see~\cite{lovasz}.  

\begin{theorem}\label{algnumrecon}
Suppose we know $\mathcal{O}(d^2+d\log{H})$ bits of precision of a complex root $\alpha$ of an unknown irreducible polynomial, $h(x)$, where the degree of $h$ is $d$ and its maximal coefficient has absolute value $\leq H$.  Algorithm~\ref{mainalg} can be used to find $h(x)$ in $\mathcal{O}(d^7 + d^5\log^2{H})$ CPU operations.
\end{theorem}

This new complexity is an improvement over the ${\rm L}^2$ algorithm which would use $\mathcal{O}(d^9 + d^7\log^2{H})$ CPU operations to reduce the same lattice.  Although, one can prove a better switch-complexity with a two-column knapsack matrix by using~\cite[Lem.\ 2]{hanrot} to bound the determinant of the lattice as $\mathcal{O}(X^2)$ and thus the potential function from~\cite{LLL} is $\mathcal{O}(X^{2d})$, leading to a switch complexity of $\mathcal{O}(d\log{X})$ (posed as an open question in~\cite[sec.\ 5.3]{stehle}).  Using this argument the complexity for ${\rm L}^2$ is reduced to $\mathcal{O}(d^8 + d^6\log^2{H})$.

\textbf{Acknowledgements.}  We thank Damien Stehl\'e, Nicolas Brisebarre, and Val\'erie Berth\'e for many helpful discussions.  Also Ivan Morel for introducing us to H-LLL.  This work was partially funded by the LaRedA project of the Agence Nationale de la Recherche, it was also supported in part by a grant from the National Science Foundation.  It was initiated while the second author was hosted by the Laboratoire d'Informatique de Robotique et de Micro\'electronique de Montpellier (LIRMM).

\end{document}